\documentclass[a4paper,leqno]{amsart}
\usepackage[nodayofweek]{datetime}
\usepackage{graphicx}
\usepackage{marginnote}
\usepackage{times,mathabx}
\usepackage[usenames,dvipsnames]{xcolor}

\usepackage{enumerate,amsmath,amssymb,stmaryrd,amscd} 
\usepackage{url,color,pb-diagram,pb-diagram,mathrsfs}
\usepackage[colorlinks=true, bookmarks=true, pdfstartview=FitH, pagebackref=true]{hyperref}
\usepackage{euscript}
\usepackage{verbatim}
\usepackage{bbm}

\numberwithin{equation}{section}

\theoremstyle{plain}

\newtheorem{thm}{Theorem}[section]
\newtheorem{theorem}[thm]{Theorem}

\newtheorem{lemma}[thm]{Lemma}

\newtheorem{proposition}[thm]{Proposition}

\theoremstyle{remark}

\theoremstyle{definition}

\newcommand{\bN}{\mathbb{N}}
\newcommand{\cF}{\mathcal{F}}
\newcommand{\cY}{\mathcal{Y}}
\newcommand{\cZ}{\mathcal{Z}}

\newcommand{\ghat}{{\widehat{g}}}
\newcommand{\uhat}{\widehat{u}}

\newcommand{\scal}{\mathrm{Scal}}
\newcommand{\definedas}{\mathrel{\raise.095ex\hbox{\rm :}\mkern-5.2mu=}}

\newcommand{\ubar}{\overline{u}}

\let\<\langle 
\let\>\rangle

\DeclareMathOperator{\vol}{Vol}

\date{14th November, 2019}

\keywords{Einstein constraint equations, non-constant mean curvature, conformal method,
Lichnerowicz equation, compact manifold, prescribed scalar curvature}

\subjclass[2000]{53C21 (Primary), 35Q75, 53C80, 83C05 (Secondary)}

\begin{document}


\author[R. Gicquaud]{Romain Gicquaud}
\address[R. Gicquaud]{Institut Denis Poisson \\ UFR Sciences et Technologie \\ Facult\'e de Tours \\ Parc de Grandmont\\ 37200 Tours \\ FRANCE}
\email{\href{mailto: R. Gicquaud <Romain.Gicquaud@lmpt.univ-tours.fr>}{romain.gicquaud@lmpt.univ-tours.fr}}

\title[On the Lichnerowicz equation]{Existence of solutions to the
Lichnerowicz equation: a new proof}

\begin{abstract}
We provide a complete study of existence and uniqueness of solutions to
the Lichnerowicz equation in general relativity with arbitrary mean curvature.

\end{abstract}

\date{\today}
\maketitle
\tableofcontents

\section{Introduction}
The Lichnerowicz equation is an elliptic equation that appears in the
construction of initial data in general relativity. In the setting of this
note, let $(M, g)$ be a compact Riemannian manifold of dimension $n > 2$,
$g \in W^{2, p}$, $p > n/2$, and assume given two functions
$\tau \in L^{2p}$ and $A \in L^{2p}$. The Lichnerowicz equation has a
positive function $\phi$ as unknown and reads
\begin{equation}\label{eqLichnerowicz}
 -\frac{4(n-1)}{n-2} \Delta \phi + \scal~\phi = -\frac{n-1}{n} \tau^2 \phi^{N-1} + \frac{A^2}{\phi^{N+1}},
\end{equation}
where $\scal$ is the scalar curvature of $g$ and
$\displaystyle N \definedas \frac{2n}{n-2}$.

We refer the reader to \cite{BartnikIsenberg, ChoquetBruhat} for an
overview of the context in which this equation appears. It has attracted
attention a couple of decades ago culminating at the
classification of constant mean curvature initial data by J. Isenberg in
\cite{Isenberg}. Recently, important efforts have been put in
constructing non-constant mean curvature initial data, see
\cite{HNT1, HNT2, MaxwellNonCMC, DahlGicquaudHumbert} and
\cite{GicquaudNgo}.

The main aim of this note is to give a short proof of
existence/non-existence of solutions to \eqref{eqLichnerowicz} in the
generic case $A \not\equiv 0$. This result is well-known to a
large extent, see e.g. \cite[Theorem 1]{MaxwellNonCMC}.
The main novelty here is that there is no need
to give separate proofs according to the sign of the Yamabe quotient of
$(M, g)$. The particular case $A \equiv 0$ is the prescribed scalar
curvature equation which is similar to the problem addressed in
\cite{Rauzy,Ouyang91,Ouyang92,Tang}, see also
\cite{DiltsMaxwell,GicquaudScalar}. We will study it in Section
\ref{secPSC}.

This paper is a byproduct of the techniques developed in
\cite{DiltsMaxwell,GicquaudScalar}.\\

The outline of this paper is as follows. In Section \ref{secLocalYamabe},
we introduce the main tool to discrimitate which function $\tau$ lead to
existence of solutions to \eqref{eqLichnerowicz}.
In Section \ref{secLich}, we study the case $A \not\equiv 0$. The main
result of this section is Theorem \ref{thmCompact} which is the main
result of the paper. Section \ref{secPSC} is devoted to the case
$A \equiv 0$ which, as we indicated before, deserves a particular
treatment.\\

\noindent\textbf{Acknowledgements:} The author is grateful to
Marie-Fran\c coise Bidaut-V\'eron for useful comments on a preliminary
version of this article.

\section{Local Yamabe invariant and first conformal eigenvalue}
\label{secLocalYamabe}
For any measurable subset $V \subset M$, we define the space
\begin{equation}\label{eqDefF}
 \cF(V) \definedas  \{u \in W^{1, 2}, u \equiv 0\text{ a.e. on } M \setminus V\}
\end{equation}
of Sobolev functions vanishing outside $V$. This set is obviously reduced
to $\{0\}$ if $V$ has Lebesgue measure zero but there are larger $V$ with
$\cF(V) = \{0\}$, see for example \cite[Chapter 6]{AdamsHedberg}. Much of this
section is adapted from \cite{GicquaudScalar}.

For any $u \in W^{1, 2}$, we set
\begin{equation}\label{eqDefG}
 G_g(u) \definedas \int_M \left[\frac{4(n-1)}{n-2} |du|^2 + \scal~u^2\right] d\mu^g
\end{equation}
We also introduce, for any $u \in W^{1, 2}$, $u \not\equiv 0$,
the Rayleigh and the Yamabe quotients:
\begin{subequations}
\begin{align}
 Q^R_g(u) & \definedas G(u) / \|u\|_{L^2}^2,\nonumber\\
 Q^Y_g(u) & \definedas G(u) / \|u\|_{L^N}^2.\nonumber
\end{align}
\end{subequations}
With these definitions at hand, we introduce the local first conformal
eigenvalue $\lambda_g(V)$ and the local Yamabe invariant $\cY_g(V)$ of
any measurable subset $V \subset M$ as follows:
\begin{subequations}
\begin{align}
 \lambda_g(V) & \definedas \inf_{u \in \cF(V) \setminus \{0\}} Q^R_g(u),\nonumber\\
 \cY_g(V) & \definedas \inf_{u \in \cF(V) \setminus \{0\}} Q^Y_g(u).\nonumber
\end{align}
\end{subequations}
From the definition of an infimum, we have $\lambda_g(V) = \cY_g(V) = \infty$
if $\cF(V)$ is reduced to $\{0\}$.

\begin{proposition}\label{propSemiContinuity}
The functional $G$ defined in \eqref{eqDefG}
is sequentially weakly lower semi-continuous on $W^{1, 2}$:
for every weakly converging sequence $(u_k)_k$,
$\displaystyle
 u_k \rightharpoonup u_\infty,
$
we have
$\displaystyle
 \liminf_{k\to \infty} G(u_k) \geq G(u_\infty).
$
\end{proposition}

\begin{proof}
Note that $G_g$ can be decomposed as
\begin{equation}\label{eqG}
 G_g(u) = \frac{4(n-1)}{n-2} \int_M |du|^2 d\mu^g + \int_M \scal~u^2 d\mu^g.
\end{equation}
The first term is weakly lower semi-continuous with respect to
$u \in W^{1, 2}$ as a continuous non-negative quadratic form. For the second
one, we shall prove that, given a sequence $(u_k)_k$ in $W^{1, 2}$,
converging weakly to $u_\infty$, $u_k \rightharpoonup_{k\to \infty} u_\infty$,
we have
\[
 \int_M \scal~u_k^2 d\mu^g \to \int_M \scal~u_\infty^2 d\mu^g.
\]
To make the notation less cluttered, we denote the second term in
\eqref{eqG} as $S(u)$:
\[
 S(u) \definedas \int_M \scal~u^2 d\mu^g.
\]
Assume by contradiction that $(S(u_k))_k$ does not converge to
$S(u_\infty)$, there
exists an $\epsilon > 0$ such that, for an infinite number of integers
$k$, we have
\begin{equation}\label{eqH}
 |S(u_k) - S(u_\infty)| > \epsilon.
\end{equation}
Without loss of generality, we can assume that \eqref{eqH} holds for all
integer $k$ and also that $(u_k)_k$ converges strongly in $L^2$ to some
$\ubar_\infty \in L^2$ since the embedding $W^{1, 2} \hookrightarrow L^2$
is compact. Then we have $u_\infty \equiv \ubar_\infty$ a.e. Indeed, the
linear form
\[
 u \mapsto \int_M u (u_\infty - \ubar_\infty) d\mu^g
\]
is (strongly) continuous for the $L^2$-topology and, hence, for the
$W^{1, 2}$-topology. As a consequence,
\[
 \int_M u_\infty (u_\infty - \ubar_\infty) d\mu^g = \lim_{k \to \infty} \int_M u_k (u_\infty - \ubar_\infty) d\mu^g = \int_M \ubar_\infty (u_\infty - \ubar_\infty) d\mu^g,
\]
where the first equality holds by the $W^{1, 2}$-weak convergence of
$(u_k)_k$ to $u_\infty$ and the second one by the $L^2$-strong
convergence of $(u_k)_k$ to $\ubar_\infty$. Subtracting both equalities,
we get
\[
 \int_M |u_\infty - \ubar_\infty|^2 d\mu^g = 0,
\]
which proves that $u_\infty \equiv \ubar_\infty$ a.e.
Finally note that, since $(u_k)_k$ is weakly convergent in $W^{1, 2}$,
it is bounded and thus (by interpolation) converges in all $L^q$ spaces,
$q \in [2, N)$. Since $\scal \in L^p$, $p > n/2$, letting $q$ be such
that
$\displaystyle
  1 = \frac{1}{p} + \frac{2}{q}
$,
we have $q \in [2, N)$ and, by H\"older's inequality, $S$ is a bounded
quadratic form on $L^q$. In particular $S$ is continuous on $L^q$:
\[
 S(u_k) \to S(u_\infty).
\]
This contradicts \eqref{eqG}: $S$ is sequentially weakly continuous on
$W^{1, 2}$. This ends the proof of Proposition \ref{propSemiContinuity}.
\end{proof}

In what follows, we let $s > 0$ be the largest constant so that
\begin{equation}\label{eqSobolev}
 \|u\|_{W^{1, 2}}^2 \geq s \|u\|_{L^N}^2\quad\forall u \in W^{1, 2}.
\end{equation}

\begin{proposition}\label{propYamabe}
Given any measurable set $V \subset M$, $\lambda_g(V)$ and $\cY_g(V)$ have
the same sign (i.e. they are either both positive, both negative or both
zero).
\end{proposition}

\begin{proof}
We can assume, without loss of generality, that $\cF(V) \neq \{0\}$ for
otherwise $\cY_g(V) = \lambda_g(V) = \infty$. If $\cY_g(V) < 0$, there
exists $u \in \cF(V)$ such that $G_g(u) < 0$ so $\lambda_g(V) < 0$. 
Assume now that $\cY_g(V) > 0$, then, for all
$u \in \cF(V) \setminus\{0\}$, we have
\[
 Q^R_g(u) = \frac{G_g(u)}{\|u\|_{L^2}^2} \geq \frac{G_g(u)}{\|u\|_{L^N}^2 \vol_g(V)^{2/n}} \geq \frac{\cY_g(V)}{\vol_g(V)^{2/n}}.
\]
We conclude that
\[
 \lambda_g(V) \geq \frac{\cY_g(V)}{\vol_g(V)^{2/n}} > 0.
\]
All we have to show now is that, if $\cY_g(V) = 0$, we have
$\lambda_g(V) = 0$. Assume for the rest of the proof that $\cY_g(V) = 0$.
If $\lambda_g(V)$ were negative, there would exits $u \in \cF(V)$ such
that $G_g(u) < 0$ so $\cY_g(V) \leq Q^Y_g(u) < 0$. This proves that
$\lambda_g(V) \geq 0$. Since $\cY_g(V) = 0$, there exists a sequence
of functions $u_k \in \cF(V)$ such that $Q^Y_g(u_k) \to 0$.
Without loss of generality, we can assume that $\|u_k\|_{L^N} = 1$
so $G_g(u_k) \to 0$.

Let $q$ be as in the proof of the previous proposition. Then we have
that
\begin{align*}
G(u_k)
 &\geq \frac{4(n-1)}{n-2} \|u_k\|_{W^{1, 2}}^2 - \frac{4(n-1)}{n-2} \|u_k\|_{L^2}^2 - \left\|\scal\right\|_{L^p} \|u_k\|_{L^q}^2\\
 &\geq \frac{4(n-1)}{n-2} \|u_k\|_{W^{1, 2}}^2 - \frac{4(n-1)}{n-2} \vol_g(V)^{1-2/q} \|u_k\|_{L^q}^2 - \left\|\scal\right\|_{L^p} \|u_k\|_{L^q}^2.
\end{align*}
Hence, setting
$C = \frac{4(n-1)}{n-2} \vol_g(V)^{1-2/q} + \left\|\scal\right\|_{L^p}$,
we arrive at
\begin{equation}\label{eqEstimate}
 G_g(u_k) + C \|u_k\|_{L^q}^2 \geq \frac{4(n-1)}{n-2} \|u_k\|_{W^{1, 2}}^2.
\end{equation}
Since $q < N$, we have that $\|u_k\|_{W^{1, 2}}$ is bounded independently
of $k$. Arguing as in the proof of the previous proposition, we can
assume that $(u_k)_k$ converges weakly in $W^{1, 2}$ and strongly in
$L^2$ to some $u_\infty \in \cF(V)$.
Combining Equation \eqref{eqEstimate} with the Sobolev estimate
\eqref{eqSobolev}, we get
\[
 G_g(u_k) + C \|u_k\|_{L^q}^2 \geq \frac{4(n-1)}{n-2} s \|u_k\|_{L^N}^2 = \frac{4(n-1)}{n-2} s.
\]
Passing to the limit as $k$ goes to infinity, we conclude that
$\|u_\infty\|_{L^q} > 0$, i.e. $u_\infty \not\equiv 0$.
By the lower semicontinuity of
$G_g$, we have $G_g(u_\infty) \leq \liminf_{k \to \infty} G_g(u_k) = 0$.
Since $G_g(u_\infty) \geq 0$, we have $G_g(u_\infty) = 0$. We have proven
that
\[
 0 \leq \lambda_g(V) \leq Q^R_g(u_\infty) = 0,
\]
i.e. $\lambda_g(V) = 0$. This concludes the proof of the fact that
$\cY_g(V)$ and $\lambda_g(V)$ have the same sign.
\end{proof}

The reason why it is more convenient to work with $\cY_g(V)$ than with
$\lambda_g(V)$ is given by the following proposition.

\begin{proposition}\label{propConfInvariance}
Assume that $g$ and $h$ are two conformally related metrics,
$h = \phi^{N-2} g$, for some positive function $\phi \in W^{2, p}$. Then
for any measurable $V$ we have
\[
 \cY_g(V) = \cY_h(V).
\]
\end{proposition}

\begin{proof}
The proof is a simple calculation. Given any $u \in W^{1, 2}$, we
have
\begin{align*}
G_h(u)
 &= \int_M \left[\frac{4(n-1)}{n-2} |du|_h^2 + \scal^h~u^2\right] d\mu^h\\
 &= \int_M \left[\frac{4(n-1)}{n-2} \phi^{2-N} |du|_g^2 + \left(-\frac{4(n-1)}{n-2} \Delta^g \phi + \scal^g~\phi\right) \phi^{1-N} u^2\right] \phi^N d\mu^g\\
 &= \int_M \left[\frac{4(n-1)}{n-2} \phi^2 |du|_g^2 + \left(-\frac{4(n-1)}{n-2} \Delta^g \phi + \scal^g~\phi\right) \phi u^2\right] d\mu^g\\
 &= \int_M \left[\frac{4(n-1)}{n-2} \left(\phi^2 |du|_g^2 - (\phi \Delta^g \phi) u^2\right) + \scal^g~(\phi u)^2\right] d\mu^g\\
 &= \int_M \left[\frac{4(n-1)}{n-2} \left(\phi^2 |du|_g^2 + \<d\phi, d(\phi u^2)\>_g\right) + \scal^g~(\phi u)^2\right] d\mu^g\\
 &= \int_M \left[\frac{4(n-1)}{n-2} \left(\phi^2 |du|_g^2 + u^2 |d\phi|^2_g + 2 \<\phi d\phi, u du\>_g\right) + \scal^g~(\phi u)^2\right] d\mu^g\\
 &= \int_M \left[\frac{4(n-1)}{n-2} |d(\phi u)|_g^2 + \scal^g~(\phi u)^2\right] d\mu^g\\
 &= G_g(\phi u).
\end{align*}
Similarly,
\[
 \|u\|_{L^N_h} = \left(\int_M u^N d\mu^h\right)^{1/N} = \left(\int_M u^N \phi^N d\mu^g\right)^{1/N} = \|\phi u\|_{L^N_g}.
\]
So
\[
 Q^\cY_h(u) = Q^\cY_g(\phi u).
\]
Since $\phi$ is bounded away from zero, multiplication by $\phi$ defines
an automorphism of $\cF(V)$. Hence,
\[
 \cY_g(V) = \inf_{u \in \cF(V)} Q^\cY_g(u) = \inf_{u \in \cF(V)} Q^\cY_g(\phi u) = \inf_{u \in \cF(V)} Q^\cY_h(u) = \cY_h(V).
\]
\end{proof}

\section{Existence of solutions to the Lichnerowicz equation}
\label{secLich}

\begin{theorem}\label{thmCompact}
Let $(M, g)$ be a compact Riemannian manifold with $g \in W^{2, p}$,
$p > n/2$. Assume that $\tau \in L^{2p}$ is given. Then the following
statements are equivalent:
\begin{enumerate}
\item\label{itAllA} There exists a solution to \eqref{eqLichnerowicz}
for all $A \in L^{2p}$, $A \not\equiv 0$
 \item\label{itSomeA} There exists a solution to \eqref{eqLichnerowicz} for at least one
$A \in L^{2p}$, $A \not\equiv 0$,
 \item\label{itYamabe} The set $Z = \tau^{-1}(0)$ satisfies $\cY_g(Z) > 0$.
\end{enumerate}
Further, the solution to \eqref{eqLichnerowicz}, when it exists, is
unique unless $\cY_g(M) = 0$ and $\tau, A \equiv 0$ for which all solutions
are proportional one to another.
\end{theorem}

It should be noted that the theorem can be applied in particular when $\cZ$
has zero Lebesgue measure. This is the case if $\tau$ never vanishes or
if $0$ is a regular value for $\tau$.

This theorem reproduces results from \cite{HNT1, HNT2,MaxwellRoughCompact,MaxwellNonCMC}
and references therein (see also \cite{GicquaudSmallTT})
in which several proofs are given according to the sign of $\cY_g(M)$ and
the nullity of $\tau$. The main novelty is that the proof establishes a
direct link between existence of solutions to the Lichnerowicz equation
and the fact that $\cY_g(Z) > 0$. We first state a lemma:

\begin{lemma}\label{lmEigenvalue}
Under the assumptions of the theorem, if $\cY_g(Z) > 0$, there exists
a constant $K > 0$ such that the operator
\[
 u \mapsto -\frac{4(n-1)}{n-2} \Delta u + \scal~u + K \frac{n-1}{n} \tau^2 u
\]
has positive first eigenvalue.
\end{lemma}

\begin{proof}
Assume by contradiction that for all $k \in \bN$, the first eigenvalue of
\[
 L_k: u \mapsto -\frac{4(n-1)}{n-2} \Delta u + \scal~u + k \frac{n-1}{n} \tau^2 u
\]
is non-positive. We denote it by $\lambda_k$ and let $u_k \in W^{2, p/2}$
be the first eigenfunction normalized so that $u_k \geq 0$ and
$\|u_k\|_{L^2} = 1$. The sequence $(\lambda_k)_k$ is increasing since
\begin{align*}
\lambda_{k+1}
 &= \int_M u_{k+1} L_{k+1} u_{k+1} d\mu^g\\
 &= \int_M u_{k+1} L_k u_{k+1} d\mu^g + \int_M \frac{n-1}{n} \tau^2 u_{k+1}^2\\
 &\geq  \int_M u_{k+1} L_k u_{k+1} d\mu^g\\
 &\geq \lambda_k.
\end{align*}
We claim that the sequence $(u_k)_k$ is bounded in $W^{1, 2}$.
Indeed, we have, using the H\"older inequality:
\begin{align*}
0
& \geq \int_M \left[\frac{4(n-1)}{n-2} |du_k|^2 + \scal~u_k^2\right] d\mu^g\\
& \geq \frac{4(n-1)}{n-2} \int_M |du_k|^2 d\mu^g - \|\scal\|_{L^p} \|u_k\|_{L^N}^{\frac{n}{p}} \|u_k\|_{L^2}^{2-\frac{n}{p}} \\
& \geq \frac{4(n-1)}{n-2} \int_M |du_k|^2 d\mu^g - 2\|\scal\|_{L^p} \left[\frac{n}{p} \epsilon \|u_k\|_{L^N}^2 + \frac{2p-n}{2p} \|u_k\|_{L^2}^2 \epsilon^{-n/(2p-n)}\right]\\
& \geq \frac{4(n-1)}{n-2} \int_M |du_k|^2 d\mu^g - \|\scal\|_{L^p} \left[\frac{2n}{sp} \epsilon \|u_k\|_{W^{1, 2}}^2 + \frac{2p-n}{p} \|u_k\|_{L^2}^2 \epsilon^{-n/(2p-n)}\right],
\end{align*}
where we used the $\epsilon$-Young inequality and the Sobolev inequality
\eqref{eqSobolev}. Assuming that $\scal \not\equiv 0$ (if $\scal \equiv 0$
the argument is simpler), we choose $\epsilon$ such that 
\[
 \|\scal\|_{L^p} \frac{2n}{sp} \epsilon = \frac{2(n-1)}{n-2},
\]
so
\[
 0 \geq \frac{2(n-1)}{n-2} \int_M |du_k|^2 d\mu^g - C \|u_k\|_{L^2}^2,
\]
for some explicit constant  $C = C(n, s, p, \|\scal\|_{L^p})$. Since
$\|u_k\|_{L^2} = 1$, this proves the claim that $(u_k)_k$ is bounded in
$L^2$.

From Rellich theorem, we now extract a subsequence $(k_i)_i$ of $k$ such
that
\[
 u_{k_i} \to u_\infty\text{ in } L^2
\]
for some $u_\infty \in W^{1, 2}$. In particular, $\|u_\infty\|_{L^2} = 1$.
We can also assume that
\[
 u_{k_i} \rightharpoonup u_\infty\text{ in } W^{1, 2}.
\]
We claim that $u_\infty \equiv 0$ a.e. on $M \setminus Z$. Otherwise,
\[
 \int_M \tau^2 u_{k_i}^2 d\mu^g \to \int_M \tau^2 u_\infty^2 d\mu^g \not= 0,
\]
so
\begin{align*}
\lambda_{k_i}
&= \int_M u_{k_i} L_{k_i} u_{k_i} d\mu^g\\
&= \int_M u_{k_i} L_0 u_{k_i} d\mu^g + k_i \frac{n-1}{n} \int_M \tau^2 u_{k_i}^2 d\mu^g\\
&\geq \lambda_0 + k_i \frac{n-1}{n} \int_M \tau^2 u_{k_i}^2 d\mu^g\\
&\to_{i \to \infty} \infty,
\end{align*}
contradicting the fact that $(\lambda_k)_k$ is bounded. Since
$\|u_\infty\|_{L^2} = 1$ and belongs to $\cF(Z)$, we have a contradiction
if $\cF(Z) = \{0\}$. In the case where $\cF(Z) \not= \{0\}$, we also get
a contradiction since
\[
 \lambda_{k_i} = G_g(u_{k_i}) + k \frac{n-1}{n} \int_M \tau^2 u_{k_i}^2 \geq G_g(u_{k_i}),
\]
so, since $G_g$ is weakly lower semicontinuous,
\[
 \liminf_{i \to \infty} \lambda_{k_i} \geq \liminf_{i \to \infty} G_g(u_{k_i}) \geq G_g(u_\infty) \geq \lambda_g(Z) > 0.
\]
This gives the final contradiction.
\end{proof}

\begin{proof}[Proof of Theorem \ref{thmCompact}]
The statement \ref{itAllA} $\Rightarrow$ \ref{itSomeA} is obvious. We now
prove that \ref{itSomeA} $\Rightarrow$ \ref{itYamabe}. The proof is
similar to that of Proposition \ref{propConfInvariance}. If
$\cF(Z) = \{0\}$, Statement \ref{itYamabe} is satisfied since
$\cY_g(Z) = \infty$. Otherwise, assume given $A \in L^{2p}$ and
$\phi \in W^{2, p}$ satisfying \eqref{eqLichnerowicz}. We set
$\ghat = \phi^{N-2}g$ and $\uhat = u \phi^{-1}$. For all $u \in \cF(Z)$, we have
\begin{align*}
G_g(u)
 &= G_g\left(\phi \uhat\right)\\
 &= \int_M \left[\frac{4(n-1)}{n-2} \left(\phi^2 \left|d \uhat\right|^2_g + \left\<\phi d\phi, d(\uhat^2)\right\>_g + \uhat^2 \left|d \phi\right|^2_g\right) + \scal~\phi^2 \uhat^2\right] d\mu^g\\
 &= \int_M \left[\frac{4(n-1)}{n-2} \left(\phi^2 \left|d \uhat\right|^2_g - (\phi \Delta \phi) \uhat^2\right) + \scal~\phi^2 \uhat^2\right] d\mu^g\\
 &= \int_M \left[\frac{4(n-1)}{n-2} \phi^2 \left|d \uhat\right|^2_g + \left(\frac{A^2}{\phi^N} - \frac{n-1}{n} \tau^2 \phi^N\right) \uhat^2\right] d\mu^g\\
 &= \int_M \left[\frac{4(n-1)}{n-2} \left|d \uhat\right|^2_{\ghat} + \left(\frac{A^2}{\phi^{2N}} - \frac{n-1}{n} \tau^2\right) \uhat^2\right] d\mu^{\ghat}\\
 &\geq \int_M \left[\frac{4(n-1)}{n-2} \left|d \uhat\right|^2_{\ghat} + \frac{A^2}{\phi^{2N}} \uhat^2\right] d\mu^{\ghat} \qquad\text{(since $\uhat \in \cF(Z)$)}.
\end{align*}
This immediately rules out the possibility that $\cY_g(Z) < 0$ since
$G_g(u) \geq 0$ for all $u \in \cF(Z)$. Assume next that $M\setminus Z$
has positive Lebesgue measure. Then, $\uhat \equiv 0$ on $M\setminus Z$.
As a consequence, from the Poincar\'e inequality, there is a constant
$\mu = \mu(g, \tau)$ so that
\[
 G_g(u) \geq \frac{4(n-1)}{n-2} \int_M \left|d \uhat\right|^2_{\ghat} d\mu^{\ghat}\geq \mu \|\uhat\|_{W^{1, 2}}^2,
\]
(see e.g. \cite[Lemma 7.16]{GilbargTrudinger})
and, hence, from the Sobolev embedding theorem,
\[
 G_g(u) \geq s\mu \|\uhat\|_{L^N}^2.
\]
This proves that
\[
 \cY_g(Z) =  s \mu > 0.
\]
The only remaining possibility is that $\tau \equiv 0$ a.e. that is to
say $Z = M$ and $\cY_g(M) = 0$. From the proof of Proposition
\ref{propYamabe}, there exists a function $u_\infty \geq 0$,
$u_\infty \not\equiv 0$ so that $G_g(u_\infty) = 0$. From the inequality
\[
 G_g(u_\infty) \geq \frac{4(n-1)}{n-2} \int_M \left|d \uhat_\infty\right|^2_{\ghat} d\mu^{\ghat},
\]
we have $d\uhat_\infty \equiv 0$: $\uhat_\infty$ is a constant function.
This gives a contradiction since
\[
 0 = G_g(u_\infty) = \int_M \frac{A^2}{\phi^{2N}} \uhat_\infty^2 d\mu^{\ghat} > 0.
\]

We finally prove that \ref{itYamabe} $\Rightarrow$ \ref{itAllA}. The
proof goes as usual by the sub- and super-solution method
(see e.g. \cite[Chapter 14]{Taylor3}). Let $K$ be as
in the statement of Lemma \ref{lmEigenvalue}. We let $u$ denote the
solution to
\begin{equation}\label{eqU}
 -\frac{4(n-1)}{n-2} \Delta u + \scal~u + K \frac{n-1}{n} \tau^2 u = A^2.
\end{equation}
Since the operator on the left hand side is positive, its Green function
is positive, so $u \in W^{2, p}$ is also positive (note that $u$ is
H\"older continuous). We set
\[
\left\lbrace
\begin{aligned}
u_+ = &\lambda_+ u,\\
u_- = &\lambda_- u
\end{aligned}
\right.
\]
for some positive constants $\lambda_\pm$ to be chosen later. We want
$u_+$ to be a super-solution to the Lichnerowicz equation
\eqref{eqLichnerowicz}, i.e. $u_+$ has to satisfy
\[
 -\frac{4(n-1)}{n-2} \Delta u_+ + \scal~u_+ + \frac{n-1}{n} \tau^2 u_+^{N-1} \geq \frac{A^2}{u_+^{N+1}}.
\]
From Equation \eqref{eqU}, this is equivalent to
\[
 \frac{n-1}{n}\tau^2 \left(\lambda_+^{N-1} u^{N-1} - K \lambda_+ u\right) + \lambda_+ A^2 \geq \frac{A^2}{\lambda_+^{N+1} u^{N+1}}.
\]
This inequality holds true if both the following inequalities are fulfilled:
\[
\left\lbrace
\begin{aligned}
\lambda_+^{N-2} u^{N-2} &\geq K,\\
\lambda_+^{N+2} &\geq u^{-N-1}.
\end{aligned}
\right.
\]
Since $u$ is bounded from above and away from zero, they are true for
large enough $\lambda_+$. Calculations for the sub-solution are similar:
if $\lambda_-$ is a small enough positive constant $u_-$ is a sub-solution
to the Lichnerowicz equation \eqref{eqLichnerowicz}.
By the sub- and super-solution argument, we get existence of $u \in W^{2, p}$
solving \eqref{eqLichnerowicz}. Uniqueness of $u$ will be proven in the
next proposition.
\end{proof}

\begin{proposition}\label{propUniqueness}
Let $(M, g)$ be a compact Riemannian manifold with $g \in W^{2, p}$,
$p > n/2$. Let $\tau, A \in L^{2p}$ be two given functions. Assume given
two positive functions $\phi_1, \phi_2 \in W^{2, p}$ solving the Lichnerowicz
equation \eqref{eqLichnerowicz}.
\begin{itemize}
 \item If $\tau \not\equiv 0$ or $A \not\equiv 0$, we have $\phi_1 \equiv \phi_2$,
 \item If $\tau, A \equiv 0$, $\phi_1$ and $\phi_2$ are proportional.
\end{itemize}
\end{proposition}

\begin{proof}
The proof of this fact is well known, we refer the reader e.g. to
\cite[Proposition 2]{MaxwellNonCMC} or to \cite{DahlGicquaudHumbert}.
We present here the argument from \cite{BrezisOswald}.
Since $\phi_1$ and $\phi_2$ are bounded from below, we have that
$\phi_1^2/\phi_2$ and $\phi_2^2/\phi_1$ both belong to $W^{2, p}$.
By an integration by parts and some routine calculations, we have
\begin{equation}\label{eqBrezis}
\begin{aligned}
&\int_M \left(-\frac{\Delta \phi_1}{\phi_1} + \frac{\Delta \phi_2}{\phi_2}\right) (\phi_1^2 - \phi_2^2) d\mu^g\\
&\qquad\qquad = \int_M \left|d\phi_1 - \frac{\phi_1}{\phi_2} d\phi_2 \right|^2 d\mu^g + \int_M \left|d\phi_2 - \frac{\phi_2}{\phi_1} d\phi_1 \right|^2 d\mu^g.
\end{aligned}
\end{equation}
If we set
\[
 f(\phi) \definedas \frac{n-2}{4(n-1)} \left[\frac{A^2}{\phi^{N+2}} - \scal - \frac{n-1}{n} \tau^2 \phi^{N-2}\right],
\]
we have
\[
 -\frac{\Delta \phi_1}{\phi_1} = f(\phi_1)\quad\text{and}\quad-\frac{\Delta \phi_2}{\phi_2} = f(\phi_2),
\]
so the identity \eqref{eqBrezis} gives
\[
\begin{aligned}
&\int_M \left[f(\phi_1) - f(\phi_2)\right] (\phi_1^2 - \phi_2^2) d\mu^g\\
&\qquad\qquad = \int_M \left|d\phi_1 - \frac{\phi_1}{\phi_2} d\phi_2 \right|^2 d\mu^g + \int_M \left|d\phi_2 - \frac{\phi_2}{\phi_1} d\phi_1 \right|^2 d\mu^g.
\end{aligned}
\]
Since $f$ is a decreasing function, we have
\[
  \left[f(\phi_1) - f(\phi_2)\right] (\phi_1^2 - \phi_2^2) \leq 0\text{ a.e.}
\]
This impose that
\[
 \int_M \left|d\phi_1 - \frac{\phi_1}{\phi_2} d\phi_2 \right|^2 d\mu^g + \int_M \left|d\phi_2 - \frac{\phi_2}{\phi_1} d\phi_1 \right|^2 d\mu^g = 0.
\]
In particular, we have
\[
 d\phi_1 - \frac{\phi_1}{\phi_2} d\phi_2 = 0\text{ a.e. } \Leftrightarrow d\left(\frac{\phi_1}{\phi_2}\right) = 0\text{ a.e.}
\]
meaning that $\phi_1$ and $\phi_2$ are proportional one another and they
are equal unless $f$ is a constant function at all points of $M$, i.e.
unless $\tau, A \equiv 0$.
\end{proof}

\section{Existence of solutions to the prescribed scalar curvature
equation}
\label{secPSC}
Our focus in this section is Equation \eqref{eqLichnerowicz} with
$A \equiv 0$, namely
\begin{equation}\label{eqPSC}
 -\frac{4(n-1)}{n-2} \Delta \phi + \scal~\phi = -f \phi^{N-1},
\end{equation}
where $f = \frac{n-1}{n} \tau^2 \geq 0$. This equation is the well-known
prescribed scalar curvature equation (see e.g. \cite{Aubin} for an
introduction). The aim of this section is to give a full proof of Theorem
\ref{thmPSC} with an argument that is simpler than the one in
\cite{Rauzy,DiltsMaxwell}, following the lines of \cite{GicquaudScalar}.
One difficulty in the study of Equation \eqref{eqPSC} is to show that
$\phi \not\equiv 0$ since $\phi \equiv 0$ is a trivial solution to
\eqref{eqPSC}. This is overcome by studying the asymptotics of $\phi$ in
the non-compact case while here the argument has to be different. The
theorm we prove is the following:

\begin{theorem}\label{thmPSC}
Let $(M, g)$ be a compact Riemannian manifold with $g \in W^{2, p/2}$,
$p > n$. Assume that $f \in L^p$, $f \geq 0$, $f \not\equiv 0$, is given.
Then the following statements are equivalent:
\begin{enumerate}
 \item\label{itPSC} There exists a positive solution $\phi \in W^{2, p}$
to \eqref{eqPSC},
 \item\label{itY} We have $\cY_g(M) < 0$ and the set $Z = f^{-1}(0)$
satisfies $\cY_g(Z) > 0$.
\end{enumerate}
Further, the solution to \eqref{eqPSC}, when it exists, is unique.
\end{theorem}

The proof of \ref{itPSC}$\Rightarrow \cY_g(Z) > 0$ is entirely similar
to the one given in the proof of Theorem \ref{thmCompact} so we omit it.
Note also that the metric $h \definedas \phi^{N-2} g$ has scalar
curvature $-f$ so
\[
 \cY_g(M) = \cY_h(M) \leq Q^Y_h(1) = \frac{G_h(1)}{\vol_h(M)^{2/N}} < 0.
\]
The proof of the converse implication will occupy the remaining of this
note. We first prove it assuming that $f \in L^\infty$ and deduce the
general case from this particular case.

We introduce the functional $F$ defined for all $\phi \in W^{1, 2}$ by
\begin{equation}\label{eqDefFunc}
 F(\phi) \definedas \int_M \left[\frac{4(n-1)}{n-2} |d\phi|^2 + \scal~\phi^2 + \frac{2}{N} f |\phi|^N\right]d\mu^g
\end{equation}
Note that the assumption that $f \in L^\infty$ is required in order to
ensure that
\[
 I(\phi) = \int_M f |\phi|^N d\mu^g < \infty
\]
for all $\phi \in W^{1, 2}$. Note that $\phi \mapsto I(\phi)$ is
continuous for the strong topology and convex since $f \geq 0$.
In particular, it is weakly lower semi-continuous. From Proposition
\ref{propSemiContinuity}, we conclude that $F$ is sequentially weakly
lower semi-continuous.

We now show that $F$ is coercive. This will imply the existence of a
minimizer for $F$. The proof is similar (yet simpler) than the one given
in \cite[Proposition 4.8]{GicquaudScalar}.

\begin{lemma}\label{lmCoercivity}
Assume that \ref{itY} in Theorem \ref{thmPSC} is satisfied, then the
functional $F$ is coercive.
\end{lemma}

\begin{proof}
We assume, by contradiction, that there exists a constant $B > 0$ and a
sequence of elements $u_k \in W^{1, 2}$ such that, for all $k$,
$F(u_k) \leq B$ while $\|u_k\|_{W^{1, 2}} \to \infty$.

We first remark that $F(|u_k|) = F(u_k)$ so, upon replacing $u_k$ by
$|u_k|$, we can suppose that $u_k \geq 0$. Let $q$ be as in the proof of
Proposition \ref{propSemiContinuity}. We have
\begin{align*}
\frac{4(n-1)}{n-2} \|u_k\|_{W^{1, 2}}^2
 &\leq F(u_k) + \int_M \left(\frac{4(n-1)}{n-2} + \scal\right) u_k^2 d\mu^g\\
 &\leq B + \left(\vol_g(M)^{1-2/q} + \left\|\scal\right\|_{L^p}\right) \|u_k\|_{L^q}^2.
\end{align*}
This proves that $\|u_k\|_{L^q} \to \infty$ and that
$\|u_k\|_{W^{1, 2}} \lesssim \|u_k\|_{L^q}$. We set
$\gamma_k = \|u_k\|_{L^q}$ and $v_k \definedas \gamma_k^{-1} u_k$ so that
the sequence $(v_k)_k$ is bounded in $W^{1, 2}$ and satisfies
$\|v_k\|_{L^q} = 1$. We can assume, without loss of generality, that
$v_k$ converges weakly in $W^{1, 2}$ and strongly in $L^q$ to some
$v \in W^{1, 2}$. Since $\|v\|_{L^q} = 1$, we have $v \not\equiv 0$.

We now claim that $v \in \cF(Z)$. Indeed, we have
\begin{equation}\label{eqSupport}
 B \geq F(u_k) = \gamma_k^2 G(v_k) + \gamma_k^N \int_M f v_k^N d\mu^g = \gamma_k^N \left(\int_M f v_k^N d\mu^g + o(1)\right).
\end{equation}
If we were able to prove that
\begin{equation}\label{eqCv}
 \int_M f v_k^N d\mu^g \to \int_M f v^N d\mu^g,
\end{equation}
we would immediately conclude that
\[
 \int_M f v^N d\mu^g = 0.
\]
Yet, convergence of $(v_k)_k$ to $v$ is so weak that proving that
\eqref{eqCv} (if true) holds is delicate. We bypass this issue by the
following argument. Assume, by contradiction, that $v \not\in \cF(Z)$,
then there exist a set $W \subset M\setminus Z$ with positive measure
and an $\epsilon > 0$ such that $v \geq \epsilon \mathbbm{1}_W$ a.e. (here
$\mathbbm{1}_W$ is the indicator function of $W$). Then,
\[
 \int_M f v^q d\mu^g \geq \epsilon^q \int_M f \mathbbm{1}_W^q d\mu^g = \epsilon^q \int_W f d\mu^g > 0.
\]
As a consequence, we have, for $k$ large enough,
\[
 \int_M f v_k^q d\mu^g \geq \frac{\epsilon^q}{2} \int_W f d\mu^g.
\]
From H\"older's inequality, we have
\[
 \left(\int_M f v_k^N d\mu^g\right)^{q/N} \left(\int_M f d\mu^g\right)^{1-q/N} \geq \int_M f v_k^q d\mu^g \geq \frac{\epsilon^q}{2} \int_W f d\mu^g.
\]
This shows that
$\displaystyle
\int_M f v_k^N d\mu^g
$
is bounded from below by a positive constant. This yields a contradiction
with \eqref{eqSupport}. As a consequence, we have $v \in \cF(Z)$.

Due to our assumption on $Z$, we have
$G_g(v) \geq \cY_g(V) \|v\|_{L^N}^{2/N} > 0$. So
\[
 \liminf_{k\to \infty} G_g(v_k) \geq G_g(v) > 0.
\]
In particular, we have
\[
 \liminf_{k\to \infty} F(u_k) \geq \liminf_{k\to \infty} G_g(u_k) = \liminf_{k\to \infty} \gamma_k^2 G_g(v_k) = \infty.
\]
This contradicts the assumption $F(u_k) \leq B$. 
\end{proof}

We have now all the ingredients to conclude that $F$ admits a minimizer
$\phi$. Since $F(|\phi|) = F(\phi)$, we can assume, without loss of
generality, that $\phi \geq 0$. $\phi$ is then a solution in a weak sense
to \eqref{eqPSC}. By elliptic regularity, we conclude that $\phi \in W^{2, p}$
and by Harnack's inequality that $\phi > 0$ provided $\phi\not\equiv 0$.

We rule out the possibility that $\phi \equiv 0$ as follows. Since
$\cY_g(M) < 0$, there exists $w \in W^{1, 2}$ such that $G_g(w) < 0$. For
any $\lambda > 0$ we have
\[
 F(\lambda w) = \lambda^2 G_g(w) + \lambda^N I(w).
\]
In particular, if $\lambda$ is small enough we have $F(\lambda w) < 0$.
This shows that the zero function is not a global mimimum of $F$. This
forces $\phi \not\equiv 0$.

Uniqueness of $\phi$ is obtained by applying Proposition
\ref{propUniqueness}.

We now need to get rid of the assumption $f \in L^\infty$. For all
$k>0$, we set $f_k \definedas \min\{f, k\} \in L^\infty$. Let $\phi_k$
denote the solution to \eqref{eqPSC} with $f$ replaced by $f_k$. Note
that the zero set of $f_k$ is the same as that of $f$ so the preceding
construction applies.
It follows from the maximum principle that $\phi_{k+1} \leq \phi_k$ for
all $k > 0$ (the argument is similar to the one in the proof of
Proposition \ref{propUniqueness}). Since
$\phi_1 \in W^{2, p} \subset L^\infty$, the sequence $f_k \phi_k^{N-1}$
is uniformly bounded in $L^p$. Hence, from elliptic regularity,
the sequence $(\phi_k)_k$ is bounded in $W^{2, p}$. By the compactness
of the embedding $W^{2, p} \hookrightarrow L^\infty$ together with
elliptic regularity, there exists a subsequence $(\phi_{\theta(k)})_k$ of
$(\phi_k)_k$ that converges to some $\phi \in W^{2, p}$, $\phi \geq 0$
solving \eqref{eqPSC}. Note that, from Dini's theorem, $(\phi_k)_k$ converges
in $L^\infty$ to $\phi$. All we need to do is to exclude that $\phi \equiv 0$.

This can be done as follows. Let $w \in W^{1, 2}$ be, as before, such
that $G_g(w) < 0$. Since $W^{2, p}$ is dense in $W^{1, 2}$, we can assume
that $w \in W^{2, p} \subset L^\infty$. As before, considering
$u = \lambda w$ in the functional \eqref{eqDefFunc}, we get existence of
$v$ such that $F(v) < 0$. Set
\[
 F_k(u) \definedas \int_M \left[\frac{4(n-1)}{n-2} |d\phi|^2 + \scal~\phi^2 + \frac{2}{N} f_k |\phi|^N\right]d\mu^g
\]
So we have $F_k(\phi_k) \leq F_k(v) \leq F(v) < 0$. Now remark that
$F_k(\phi_k) \to_{k \to \infty} F(\phi)$. This forces $F(\phi) \leq F(v) < 0$
which shows that $\phi \not\equiv 0$. By construction $\phi \geq 0$ and
from Harnack's inequality, we have $\phi > 0$. This ends the proof of
Theorem \ref{thmPSC}. Uniqueness is obtained from Proposition
\ref{propUniqueness}.

\providecommand{\bysame}{\leavevmode\hbox to3em{\hrulefill}\thinspace}
\providecommand{\MR}{\relax\ifhmode\unskip\space\fi MR }
\providecommand{\MRhref}[2]{%
  \href{http://www.ams.org/mathscinet-getitem?mr=#1}{#2}
}
\providecommand{\href}[2]{#2}

\end{document}